\documentclass{jpc}

\usepackage[english]{babel}
\usepackage{mathtools}
\usepackage[labelfont=bf,width=0.95\textwidth]{caption}
\usepackage{subcaption}
\usepackage{amsmath}
\usepackage{amsfonts}
\usepackage{amssymb}
\usepackage{cite}
\usepackage{graphicx}
\usepackage{thm-restate}
\usepackage[usenames,dvipsnames,table]{xcolor}
\usepackage[ruled,vlined]{algorithm2e}
\usepackage{balance}
\usepackage{xspace}
\usepackage[capitalize,noabbrev]{cleveref}
\usepackage{outlines}
\usepackage{outlines}
\usepackage{multicol}
\usepackage{multirow}
\usepackage{bm,bbm}
\usepackage{array}
\usepackage{tabularx}
\usepackage{dcolumn}
	\newcolumntype{d}[1]{D{.}{.}{#1|}}
\usepackage{booktabs}
\usepackage{soul}
\usepackage[group-separator={,}]{siunitx}
\usepackage{mathbbol}
\usepackage{adjustbox}
\usepackage{framed}
\usepackage{listings}
\usepackage{pythonhighlight}
\usepackage[bottom]{footmisc}

\DeclareSymbolFontAlphabet{\mathbb}{AMSb}
\DeclareSymbolFontAlphabet{\mathbbl}{bbold}
\SetNlSty{textbf}{(}{)}
\SetKwFor{For}{}{}{}

\newcommand{\set}[1]{\{#1\}}
\DeclareMathOperator*{\argmin}{argmin}

\DeclarePairedDelimiter\floor{\lfloor}{\rfloor}

\newtheorem{remark}{Remark}
\newtheorem{theorem}{Theorem}

\newtheorem{definition}{Definition}

\newtheorem{proposition}{Proposition}

\newtheorem{example}{Example}

\newtheorem{fct}{Fact}

\renewcommand{\revision}[1]{#1}

\newcommand{\norm}[1]{\left\lVert#1\right\rVert}
\newcommand{\eat}[1]{}

\newcommand{\pgm}{\texttt{Private-PGM}\xspace}
\newcommand{\marginal}[1]{(\texttt{#1})}
\newcommand{\attr}[1]{\texttt{#1}}
\newcommand{\dom}{\Omega}
\newcommand{\domsp}{\dom^{\text{SPECS}}}
\newcommand{\domip}{\dom^{\text{IPUMS}}}
\newcommand{\mech}{\texttt{NIST-MST}\xspace}
\newcommand{\mst}{\texttt{MST}\xspace}
\newcommand{\INCWAGE}{\attr{INCWAGE}\xspace}
\newcommand{\INCWAGEA}{\attr{INCWAGE\textsubscript{A}}\xspace}
\newcommand{\INCWAGEB}{\attr{INCWAGE\textsubscript{B}}\xspace}
\newcommand{\algsize}{\footnotesize}

\title[Title]{Winning the NIST Contest: A scalable and general approach to differentially private synthetic data}

\author{Ryan McKenna}   
\address{College of Information \& Computer Sciences, The University of Massachusets, Amherst, MA 10002}
\email{rmckenna@cs.umass.edu}  

\author{Gerome Miklau}   
\address{College of Information \& Computer Sciences, The University of Massachusets, Amherst, MA 10002}
\email{miklau@cs.umass.edu}  

\author{Daniel Sheldon}   
\address{College of Information \& Computer Sciences, The University of Massachusets, Amherst, MA 10002}
\email{sheldon@cs.umass.edu}  
 
 
\begin{document} 

\maketitle 
 
\pagestyle{plain} 

\begin{abstract} 
We propose a general approach for differentially private synthetic data generation, that consists of three steps: (1) \textbf{select} a collection of low-dimensional marginals, (2) \textbf{measure} those marginals with a noise addition mechanism, and (3) \textbf{generate} synthetic data that preserves the measured marginals well.  Central to this approach is \pgm \cite{mckenna2019graphical}, a post-processing method that is used to estimate a high-dimensional data distribution from noisy measurements of its marginals.  We present two mechanisms, \mech and \mst, that are instances of this general approach.  \mech was the winning mechanism in the 2018 NIST differential privacy synthetic data competition, and \mst is a new mechanism that can work in more general settings, while still performing comparably to \mech.  We believe our general approach should be of broad interest, and can be adopted in future mechanisms for synthetic data generation.
\end{abstract}

\section{Introduction}

Data sharing within the modern enterprise is extremely constrained by privacy concerns.
Privacy-preserving synthetic data is an appealing solution: it allows existing analytics work-
flows and machine learning methods to be used while the original data remains protected.
But recent research has shown that unless a formal privacy standard is adopted, synthetic data can violate privacy in subtle ways \cite{dinur2003revealing,garfinkel2019understanding}.  Differential privacy offers such a formalism, and the problem of differentially private synthetic data generation has therefore received considerable research attention in recent years \cite{zhang2017privbayes,chen2015differentially,zhang2019differentially,xu2017dppro,xie2018differentially,torfi2020differentially,vietri2020new,liu2016model,torkzadehmahani2019dp,charest2011can,ge2020kamino,huang2019psyndb,jordon2018pate,zhang2018differentially,tantipongpipat2019differentially,abay2018privacy,bindschaedler2017plausible,zhang2020privsyn,asghar2019differentially,li2014differentially}.

In 2018, the National Institute of Standards and Technology (NIST) highlighted the importance of this problem by organizing the  {\em Differential Privacy Synthetic Data Competition} \cite{nist}.  This competition was the first of its kind for the privacy research community, and it encouraged privacy researchers and practitioners to develop novel practical mechanisms for this task. The competition consisted of three rounds of increasing complexity. In this paper we describe \mech, the winning entry in the third and final round of the competition. Our algorithm is an instance of a general template for differentially private synthetic data generation that we believe will simplify design of future mechanisms for synthetic data. 

Our approach to differentially private synthetic data generation consists of three high-level steps, as show in \cref{fig:template}: (1) query selection, (2) query measurement and (3) synthetic data generation.  For step (1), there are various ways to approach query selection; a domain expert familiar with the data and its use cases can specify the set of queries, or they can be automatically determined by an algorithm.  The selected queries are important because they will ultimately determine the statistics for which the synthetic data preserves accuracy. For step (2), after the queries are fixed, they are measured privately with a noise-addition mechanism, in our case, with the Gaussian mechanism.  In step (3), the noisy measurements are processed through \pgm \cite{mckenna2019graphical}, a post-processing method that can estimate a high-dimensional data distribution from noisy measurements and generate synthetic data. 
 
This approach is similar in spirit to the widely studied select-measure-reconstruct paradigm for linear query answering under differential privacy \cite{zhang16privtree,yuan2016convex,li2015matrix,zhangtowards,xiao2014dpcube,qardaji2014priview,li2014data,Yaroslavtsev13Accurate,xu2013differential,qardaji2013understanding,qardaji2013differentially,yuan2012low,xu12histogram,li2012adaptive,cormode2012differentially,Acs2012compression,xiao2011differential,ding2011differentially,li2010optimizing,hay2010boosting,li2015matrix,mckenna2018optimizing}.  However, the output is now synthetic data, rather than query answers.  In addition, most existing methods from this paradigm suffer from the curse of dimensionality, and have trouble scaling to high-dimensional domains.  Our approach is simple and modular but there are three main technical challenges to using it in practice.  These are 
(1) identifying what statistics to measure about the dataset, 
(2) generating synthetic data that effectively preserves the measured statistics,
and (3) overcoming the challenges of high-dimensional domains.
Fortunately, \pgm solves problem (2) and (3) above, as long as the measured statistics only depend on the data through its low-dimensional marginals.  This allows the mechanism designer to focus on problem (1), and frees them from the burden of figuring out \emph{how} to generate synthetic data with differential privacy, allowing them instead to focus on \emph{what} statistics to measure, based on what they want the synthetic data to preserve.  Thus, we believe this approach to differentially privacy synthetic data, using \pgm, will be broadly applicable.

\begin{figure}[t]
\includegraphics[width=\textwidth]{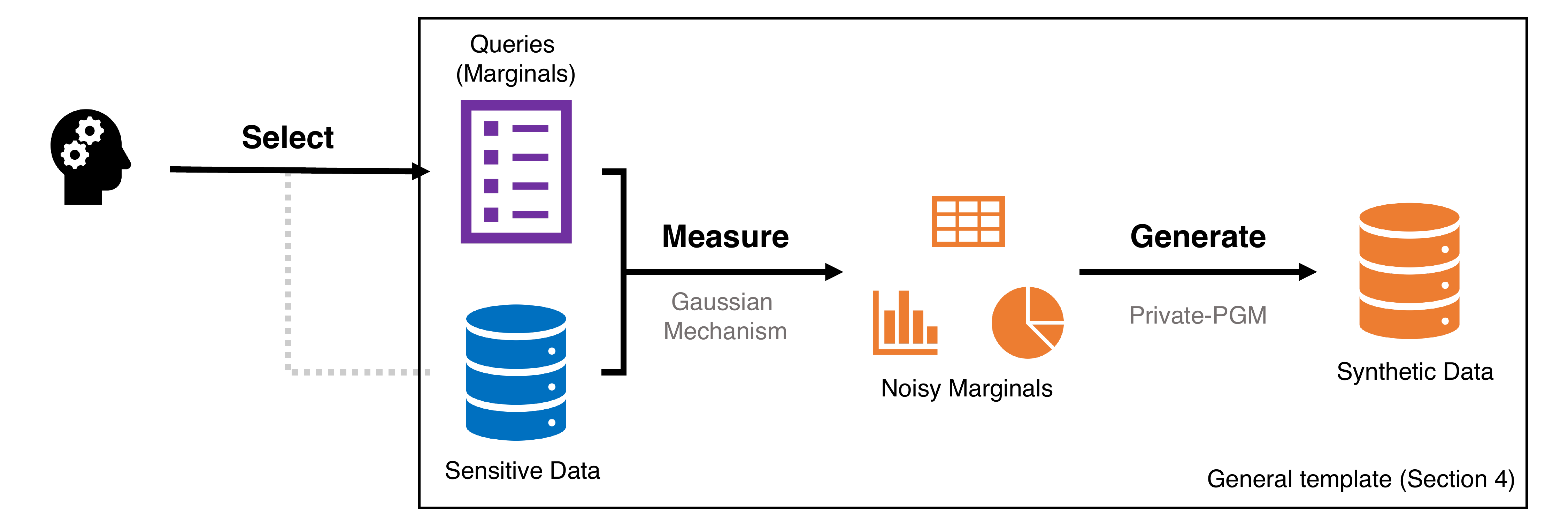}
\caption{\label{fig:template}A general template for differentially private synthetic data generation.  First, a collection of marginal queries is \textbf{selected}, either manually (e.g., by a domain expert) or automatically by an algorithm. Second, the Gaussian mechanism is used to \textbf{measure} those marginals while preserving differential privacy.  Finally, \pgm is used to post-process the noisy marginals and \textbf{generate} a synthetic dataset that respects them.}
\end{figure}

In \cref{fig:template}, there is a dashed gray line connecting the select step with the sensitive data.  This indicates that query selection may or may not depend on the sensitive data, but if it does, it must be via a differentially private mechanism.  The rules of the NIST competition permitted the use of a public provisional dataset which the \mech algorithm uses for query selection. The effectiveness of query selection relies on the similarity of the public data to the private data being synthesized. Since high quality provisional data may not always be available, we propose a variant of the algorithm, called \mst, that does not require public data and instead uses a portion of the privacy budget to select measurements.  The novelty of this algorithm is that it uses the data (privately) to select marginals to measure that support efficient synthesis in step (3). This extension leads to an algorithm that can be applied in a wider variety of settings.  We show experimentally that, without the advantage of provisional data, it nevertheless performs comparably to \mech.

This paper is organized as follows.  In \cref{sec:background}, we set up notation, state assumptions, and summarize relevant background in differential privacy.  In \cref{sec:problem}, we summarize the important aspects of the competition.  In \cref{sec:pgm}, we describe the general template from \cref{fig:template} in greater detail.  In \cref{sec:mech}, we present \mech, the winning mechanism from the competition, by building on the general template.  In \cref{sec:extensions}, we present \mst, a novel mechanism inspired by \mech that does not rely on the existence of public provisional data.  We conclude with a simple experimental evaluation and discussion of results.

\section{Background} \label{sec:background}

\def\alg{{\mathcal M}}
\def\db{\mathbf{X}}
\def\nbrs{\textrm{nbrs}}

\newcommand{\compl}[1]{{#1}^{-}}

The algorithms described in this paper take as input a dataset, assumed to be a single table, and generate a synthetic dataset satisfying $(\epsilon, \delta)$-differential privacy.  Below we provide the relevant background and notation on datasets, marginals, and differential privacy.

\subsection{Data}

The input is a dataset $D$ consisting of $m$ records, each containing potentially sensitive information about one individual.  Each record has $d$ attributes $ \mathcal{A} = \set{A_1, \dots, A_d} $, and the domain of possible values for an attribute $A_i$ is denoted by $\Omega_i$.  We assume $\Omega_i$ is finite and has size $ | \Omega_i | = n_i$.   The full domain of possible values is thus $ \Omega = \Omega_1 \times \dots \times \Omega_d$ which has size $\prod_i n_i = n $.  We use $\mathcal{D}$ to denote the set of all possible datasets, which is equal to $ \mathcal{D} = \cup_{m=0}^{\infty} \Omega^m $. 

\subsection{Marginals}

A marginal is a key statistic that captures low-dimensional structure in a high-dimensional data distribution.   
We will explain (in Section \ref{sec:problem}) that the evaluation metrics of the contest can be defined in terms of marginals computed on the dataset.  In addition, our algorithms will privately measure selected marginals and use the resulting noisy measurements to construct synthetic data.  More precisely, a marginal, for a set of attributes $C$, is a table that counts the number of occurrences of each combination of possible values for attributes $C$.

\begin{definition}[Marginal]
Let $C \subseteq \mathcal{A}$ be a subset of attributes, $\Omega_C = \prod_{i \in C} \Omega_i$, and $n_C = | \Omega_C |$.  The marginal on $C$ is a vector $\mu \in \mathbb{R}^{n_C}$, indexed by domain elements $t \in \Omega_C$, such that each entry is a count, i.e., $\mu_t = \sum_{x \in D} \mathbbm{1}[x_C = t]$.  We let $M_C : \mathcal{D} \rightarrow \mathbb{R}^{n_C}$ denote the function that computes the marginal on $C$, i.e., $ \mu = M_C(D) $.  
\end{definition}

\subsection{Differential privacy}
Differential privacy \cite{dwork2006calibrating,Dwork14Algorithmic} protects individuals by bounding the impact any one individual can have on the output of an algorithm. This is formalized using the notion of neighboring datasets.  Two datasets $D, D' \in \mathcal{D}$ are neighbors (denoted $ D \sim D'$) if $D'$ can be obtained from $D$ by adding or removing a single record.  


\begin{definition}[Differential Privacy~\cite{dwork2006calibrating}] \label{def:dp}
A randomized mechanism $\alg : \mathcal{D} \rightarrow \mathcal{R} $ satisfies $(\epsilon, \delta)$-differential privacy (DP) if for any neighboring datasets $D \sim D' \in \mathcal{D}$, and any subset of possible outputs $S \subseteq \mathcal{R}$, 
$$ \Pr[\alg(D) \in S] \leq \exp(\epsilon) \Pr[\alg(D') \in S] + \delta.$$
\end{definition}

This definition requires that, on any two neighboring input databases, the difference in the output distributions of the randomized algorithm $\mathcal{M}$ is bounded by $e^\epsilon$, except with a small failure probability $\delta$. This failure probability $\delta$ is usually assumed to be cryptographically small; in the contest it was set to $\delta \approx 2 \cdot 10^{-12}$.  The algorithms in this paper achieve differential privacy by repeated application of the Gaussian mechanism \revision{and the Exponential Mechanism}, defined below:

\begin{definition}[Gaussian Mechanism]
Let $f : \mathcal{D} \rightarrow \mathbb{R}^p $ be a vector-valued function of the input data. The Gaussian Mechanism adds i.i.d. Gaussian noise with scale $\sigma$ to $f(D)$:
$$ \alg(D) = f(D) + \mathcal{N}(0, \sigma^2 \mathbf{I}).$$
\end{definition}
\revision{
\begin{definition}[Exponential Mechanism]
Let $q : \mathcal{D} \times \mathcal{R} \rightarrow \mathbb{R}$ be quality score function and $\epsilon$ be a parameter.  Then the exponential mechanism outputs a candidate $r \in \mathcal{R}$ according to the following distribution:
\vspace{-0.5em}
$$ \Pr[\mathcal{M}(D) = r] \propto \exp{\Big( \epsilon \cdot q(D, r) \Big)} $$
\end{definition}}

To accurately analyze the privacy of multiple invocations of the Gaussian/\revision{Exponential} mechanisms (i.e., to derive the $(\epsilon, \delta)$ parameters) we use the tools of R\'enyi Differential Privacy (RDP), a variant of differential privacy so named because it uses the R\'enyi divergence in the bound on a mechanism's output distributions for neighboring inputs.

\begin{definition}[R\'enyi Differential Privacy \cite{mironov2017renyi}] 
A randomized mechanism $\alg : \mathcal{D} \rightarrow \mathcal{R}$ satisfies $(\alpha, \gamma)$-R\'enyi differential privacy (RDP) for $\alpha \geq 1$ and $\gamma \geq 0$, if for any neighboring datasets $D \sim D' \in \mathcal{D}$, we have:
$$ D_{\alpha}(\alg(D) \mid\mid \alg(D')) \leq \gamma, $$
where $D_{\alpha}(\cdot \mid \mid \cdot)$ is the R\'enyi divergence of order $\alpha$ between two probability distributions.
\end{definition} 

To analyze the privacy of the \revision{mechanisms above} under R\'enyi-DP, we define the sensitivity of a vector-valued query as follows:

\begin{definition}[Sensitivity]
Let $f : \mathcal{D} \rightarrow \mathbb{R}^p $ be a vector-valued function of the input data.  The $L_2$ sensitivity of $f$ is $\Delta_f = \max_{D \sim D'} \norm{f(D) - f(D')}_2$.  
\end{definition}

It is easy to verify that the $L_2$ sensitivity of any marginal function $M_C$ is $1$, regardless of the attributes in $C$. This is because one individual can only contribute a count of $1$ to a single cell of the output vector.  A single invocation of the Gaussian Mechanism satisfies R\'enyi-DP with parameters determined by the noise scale $\sigma$ and the sensitivity $\Delta_f$ of the function.  \revision{Similarly, a single invocation of the Exponential Mechanism also satisfies R\'enyi-DP with parameters determined by $\epsilon$ and $\Delta_q$.}

\begin{proposition}[R\'enyi-DP of the Gaussian Mechanism \cite{mironov2017renyi,feldman2018privacy}] \label{prop:rdpgauss}
The Gaussian Mechanism applied to the function $f : \mathcal{D} \rightarrow \mathbb{R}^p$ satisfies $\Big( \alpha, \alpha \frac{\Delta_f^2}{2 \sigma^2} \Big)$-RDP for all $\alpha \geq 1$.
\end{proposition}

\revision{
\begin{proposition}[R\'enyi-DP of the Exponential Mechanism \cite{mcsherry2007mechanism,cesar2021bounding}] \label{prop:expprivacy}
The Exponential Mechanism applied to the quality score function $q : \mathcal{D} \times \mathcal{R} \rightarrow \mathbb{R}$ satisfies $(2 \epsilon \Delta, 0)$-DP and $(\alpha, \alpha \frac{(2 \epsilon \Delta)^2}{8})$-RDP for all $\alpha \geq 1$, where $\Delta = \max_{r \in \mathcal{R}} \Delta_{q(\cdot,r)}$ is the maximum sensitivity of $q$. 
\end{proposition}
Note that any mechanism that is $(\alpha, \alpha \rho)$-RDP for all $\alpha \geq 1$ is also $\rho$-zCDP \cite{bun2016concentrated} and vice-versa.  
We rely on the following propositions to reason about multiple adaptive invocations of RDP mechanisms, and the translation between R\'enyi-DP and $(\epsilon, \delta)$-DP.}

\begin{proposition}[Adaptive Composition of RDP Mechanisms \cite{mironov2017renyi}] \label{prop:composition}
Let $\alg_1 : \mathcal{D} \rightarrow \mathcal{R}_1$ be $ (\alpha, \gamma_1)$-RDP and $\alg_2 : \mathcal{D} \times \mathcal{R}_1 \rightarrow \mathcal{R}_2$ be $(\alpha, \gamma_2)$-RDP.  Then the mechanism $\alg = \alg_2(D, \alg_1(D))$ is $(\alpha, \gamma_1 + \gamma_2)$-RDP.
\end{proposition}

\begin{proposition}[RDP to DP \cite{mironov2017renyi}] \label{prop:rdpdp}
If a mechanism $\alg$ satisfies $(\alpha, \gamma)$-R\'enyi differential privacy, it also satisfies 
 $\Big(\gamma + \frac{\log{(1/\delta)}}{\alpha-1}, \delta \Big)$-differential privacy for all $\delta \in (0,1]$.
\end{proposition}


\section{Competition setup} \label{sec:problem}

In this section we will summarize the format of the competition and the different components of the challenge problem.  The competition consisted of three rounds of increasing complexity, but our focus is on the third round, which built on the previous two rounds.   

\subsection{Competition Format} 

Competitors were given approximately one month to design their differentially private synthetic data mechanism.  The competition organizers provided contestants with a precise problem specification along with a ``competitor pack'', which included a provisional dataset to test and develop mechanisms, a file that contained domain information for each attribute in the dataset, a script to evaluate the quality of the synthetic data according to their custom scoring criteria, and a baseline mechanism. Each of these components will be described in detail in the subsequent sections.  
During the one-month competition period, competitors could submit the synthetic data produced by their mechanism to be scored and attain a spot on the provisional leaderboard.  This was a good way to gauge how well other competitors were doing, although it was not an authoritative source as the submissions had not been vetted to ensure they satisfied differential privacy (e.g., someone could submit the true data and get a perfect score on the provisional leaderboard). 

At the end of the competition period, competitors had to submit their source code along with a document describing the solution and a proof of privacy.   A team of experts unknown to the competitors checked the final submitted algorithms to ensure that they satisfied $(\epsilon, \delta)$-differential privacy.  This was done by checking the written description of the algorithm as well as the source code, to ensure there were no privacy violations or mistakes. After the mechanism was verified to be differentially private, its utility was evaluated on the final dataset (different from the provisional dataset) and performance results were added to a final leaderboard that would determine the ranking of solutions.   

In the subsequent subsections, we will specify the details of each component of the problem --- i.e., the dataset, the domain, and the evaluation criteria.  

\subsection{Dataset} \label{sec:dataset}

Algorithms were designed for and evaluated on data from the 1940 U.S. decennial census.   The provisional dataset contained data for one state (Colorado) and the final holdout data was for a different state (unknown at the time of the competition).  
We remark that the provisional dataset was treated as public information.  Therefore, any analysis and insights derived from it were not considered to violate privacy. However, solutions that used information from the provisional dataset too aggressively risked over-fitting, and scoring poorly on the final dataset.   
The provisional dataset contained 98 attributes and about 661 thousand records.  All attributes were discrete, taking on values from the domain $\Omega_i = \set{0, \dots, n_i-1}$.  The value of $n_i$ for each attribute $i$ was provided in a separate specs file.  The values of $n_i$ ranged from $2$ (for binary attributes like \attr{SEX}) to $10,000,000$ (for numerical attributes like \attr{INCWAGE}).  The total number of possible database rows (i.e., the full domain size) was about $ | \dom | = 5 \times 10^{205} $.  We provide a full breakdown of the domain in \cref{table:domain}.

\subsection{Evaluation metrics} \label{sec:eval}

The utility of the synthetic data was measured by how well it preserved key statistics in the ground truth data with respect to three main criteria, enumerated below.  We state below the statistics that need to be preserved to score well, but not the exact formula for calculating score.  For the precise information, please refer to the official challenge problem statement \cite{nist}.  
Note that the scores for each evaluation metric were normalized to the same range, and averaged across the three metrics (with equal weights).  Algorithms were evaluated at three privacy levels, with $\epsilon = 0.3, 1.0, 8.0$ and $\delta \approx 2 \cdot 10^{-12}$, and these scores were averaged to obtain the final score.  Computational efficiency was not taken into consideration; several of the solutions (including \mech) required up to 10 minutes or more to run. 

\begin{enumerate}
\item \textbf{3-way Marginals.}
The synthetic data was evaluated by comparing its marginals with the marginals of the true data for $100$ random triples of attributes, unknown to competitors at submission time. Therefore, a synthetic dataset $\tilde{D}$ scores well on this metric (in expectation) if $M_C(D) \approx M_C(\tilde{D})$ \emph{for all} triples $C$.  There are a total of $\binom{98}{3}=152096$ possible triples, so this evaluation criteria requires the synthetic data to preserve a large number of marginals to consistently score well.
\item \textbf{High-order conjunctions.}
The synthetic data was evaluated by looking at how well it preserved high-order conjunctions.  Probabilistically, a high-order conjunction for a set of attributes $C \subseteq \mathcal{A} $ assumes the form $\Pr[\bigwedge_{i \in C} [t_i \in S_i] \mid t \sim D]$, where $S_i \subseteq \Omega_i$ is a subset of the domain for attribute $i$.  This quantity can be expressed in terms of the marginal $\mu = M_C(D)$ via $\sum_{t \in S} \mu_t $, where $S$ is the Cartesian product of $S_i$'s, i.e., $S = \prod_{i \in C} S_i$ and $t$ is a tuple restricted to the attributes in the set $C$.  Synthetic data was evaluated on $300$ random high-order conjunctions, where $C$ is generated with a simple random sample of the attributes $\mathcal{A}$ with selection probability $0.1$, and $S_i$ is a random subset of $\Omega_i$. There are a total of $2^{98} \approx 10^{29}$ possible choices for $C$, and the expected size of $C$ is $ 0.1 \cdot 98 \approx 10$.  Even without accounting for the variability in $S_i$, it is clear that the number of statistics that need to be preserved is enormous. 
\item \textbf{Income inequality and gender wage gap.}
The synthetic data was evaluated by how well it preserved statistics relating to income inequality and gender wage gap, broken down by city.  While the precise details of this metric can be found in the official problem statement, to score well, it suffices for the synthetic data to be accurate with respect to the marginal on \marginal{SEX,CITY,INCWAGE}.  Unlike metrics (1) and (2), above, this metric is relatively easy to score well on, because it just requires preserving one marginal well.
\end{enumerate}

\section{Overview of Measurement and Inference with \pgm} \label{sec:pgm}
\begin{figure} 
\begin{python}
from private_pgm import FactoredInference
from scipy.sparse import identity
from numpy.random import normal

data = load_NIST()
queries = [("SEX","LABFORCE"), ("LABFORCE","SCHOOL")]
measurement_log = []
for c in queries:
    M_c = data.project(c).datavector()
    y_c = M_c + normal(loc=0, scale=50, size=M_c.size)
    measurement_log.append( (identity(M_c.size), y_c, 50, c) ) 
engine = FactoredInference(data.domain)
model = engine.estimate(measurements)
synth = model.synthetic_data()
\end{python}
\vspace{-3mm}\caption{\label{fig:pythonexample}A demonstration of how to generate synthetic data with \pgm using real \texttt{Python} code.  In this case, the selected marginals are \marginal{SEX,LABFORCE} and \marginal{LABFORCE,SCHOOL}.  In Lines 9-12, these marginals are measured with Gaussian noise to protect privacy.  In Lines 14-16, \pgm takes these noisy measurements as input, estimates a model, and generates synthetic data.  \revision{ The \pgm library provides a straightforward interface that allows users to quickly write end-to-end code to generate synthetic data; different statistics can be preserved by changing Line 6.}}
\end{figure}

In this section, we elaborate on the general template for a mechanism outlined in \cref{fig:template}.   Recall there are three high-level steps:
\begin{enumerate}
\item \textbf{Select.} Select a collection of marginals to measure.
\item \textbf{Measure.} Use the Gaussian mechanism to measure each marginal in the collection.
\item \textbf{Generate.} Use \pgm to estimate a data distribution from the noisy measurements and generate synthetic data that preserves the measured marginals well.
\end{enumerate}
In this section, we describe the latter two steps, which form the core of the mechanism. In the next section, we will describe the full mechanism \mech in detail, including the \emph{select} step and many other details relating specifically to the NIST contest and dataset. 
\cref{fig:pythonexample} shows how simple and modular this framework for synthetic data generation is.  
\revision{The open source \pgm library\footnote{\url{https://github.com/ryan112358/private-pgm/}} provides a simple interface to the key routines so that an end-to-end synthetic data generation mechanism can be written with very little code, allowing the modeler to focus on tailoring the procedure to the workload and domain.}
\revision{Under the hood, \pgm has thousands of lines of code, but it exposes a simple interface that is easy to use.}
In this example, there are only two selected marginal queries: \marginal{SEX,LABFORCE} and \marginal{LABFORCE,SCHOOL}, but the code can be readily modified (Line 6) to accommodate other marginal queries.  
In the rest of the section, we describe this \revision{general} approach in more detail, and give some insight into the steps described in this code snippet.  
\subsection{Measuring Marginals with the Gaussian Mechanism}
\cref{alg:measure} shows the method for measuring marginals. Given a collection of attribute subsets $\mathcal{C}$, it computes the marginal for each $C \in \mathcal{C}$ and adds i.i.d. Gaussian noise to preserve privacy.  It also accepts a weight $w_C$ for each attribute subset, which represents the relative importance of that marginal.  It collects all of these noisy measurements into a {\em measurement log}, which will be passed to \pgm for post-processing. The measurement log records, for each marginal defined by a subset of attributes, the noisy marginal query answers together with information about the weight assigned to the marginal and the magnitude of noise used to measure it. It is easy to verify the privacy properties of \cref{alg:measure}, as it is a direct application of the Gaussian mechanism on a sensitivity-$1$ quantity.\footnote{The weights are explicitly normalized so that the collection of marginals has sensitivity $1$.}

\begin{theorem} \label{thm:privacymeasure}
\cref{alg:measure} satisfies $(\alpha, \frac{\alpha}{2 \sigma^2})$-RDP for all $\alpha \geq 1$.
\end{theorem}

\begin{algorithm}[t]
\algsize{
\textbf{Input:} $D$ (sensitive dataset), $\mathcal{C}$ (a collection of attribute subsets), $w_C$ (weights for each $C \in \mathcal{C})$, $\sigma$ (noise scale) \\
\textbf{Output:} log (a list of noisy measurements together with metadata) \\
\nl Normalize weights, $w_C \leftarrow w_C / \sqrt{\sum_C w_C^2}$. \\
\nl \For{\emph{For $C \in \mathcal{C}$:}}{
\nl Calculate noisy marginal, $\tilde{\mu} = w_C M_C(D) + \mathcal{N}(0, \sigma^2 I)$ \\
\nl Append $4$-tuple $(w_C I,  \; \tilde{\mu}, \; \sigma,\; C)$ to measurement log \\
}}
\caption{ \label{alg:measure} Measure Marginals}
\end{algorithm}

\subsection{\pgm: Inference and Synthetic Data Generation}
\pgm is a general-purpose post-processing tool to infer a data distribution given noisy measurements~\cite{mckenna2019graphical}. It is compatible with measurements from a wide variety of mechanisms for discrete data, and can often improve utility at no cost to privacy. Because it infers a representation of a full data distribution, it produces query answers that are \emph{consistent} with one another, even if the noisy measurements are inconsistent. It uses a compact representation of the data distribution to avoid exponential complexity in many cases, though the size of the representation will depend on the measurements, as we describe below.


The high-level idea of \pgm is to solve an optimization problem to find a data distribution that would produce measurements close to the ones that were observed. It applies to cases when private measurements depend on the data through marginals. For example, suppose the measurements are of the form
$$y_C = Q_C M_C(D) + \xi$$
for all attributes sets $C$ in some collection $\mathcal{C}$, where $Q_C \in \mathbb{R}^{p_C \times n_C}$ is a linear transformation applied to the marginal prior to release and $\xi \in \mathbb{R}^{p_C}$ is zero-centered noise (e.g., Laplace or Gaussian) with known standard deviation. The measurements taken in \cref{alg:measure} represent the common case where $Q_C$ is just the identity matrix, so that we observe the noisy marginals directly. However, the ability to measure arbitrary linear transformations of marginals is a nice feature that is useful for some types of measurements that occur in practice.\footnote{In fact, $Q_C$ can be replaced with an arbitrary non-linear differentiable transformation, and \pgm will accept that as input as well.}  Examples of this include hierarchical measurements for answering range queries \cite{hay2010boosting,qardaji2013understanding,li2014data}, and optimized measurements for answering general linear query workloads \cite{li2010optimizing,mckenna2018optimizing}.   

Given these measurements, \pgm infers a data distribution $P$ that best explains the measurements by solving the optimization problem
\begin{equation}
  \label{eq:op}
\argmin_{P} \sum_{C \in \mathcal{C}} \norm{ Q_C M_C(P) - y_C}_2^2.
\end{equation}
The objective of this optimization problem is the negative log-likelihood of the noisy measurements under the Gaussian release mechanism, so \cref{eq:op} can be seen as a maximum likelihood estimator. We have abused notation by allowing $M_C$ to operate on a data distribution rather than a dataset; the correct interpretation is to substitute the probability vector $P$ for the contingency table representation $D$ of the dataset for computing the marginal. 
An obvious issue with the optimization problem in \cref{eq:op} is that the dimension of the decision variable $P$ is equal to the domain size $n$, which is exponential in the number of attributes, so we cannot usually solve this problem directly. The key observation of \pgm is
\begin{fct}
\label{fact:op}
\cref{eq:op} has an optimum of the form $P_\theta$, where $P_\theta$ is a graphical model with one factor for each set $C \in \mathcal{C}$ of attributes for which the mechanism measured a marginal.
\end{fct}
This allows us to solve the much lower-dimensional optimization problem
\begin{equation} \label{eq:op2}
\argmin_{\theta} \sum_{C \in \mathcal{C}} \norm{ Q_C M_C(P_\theta) - y_C}_2^2,
\end{equation}
with no loss in solution quality. The decision variable $\theta$ is the parameter vector of the graphical model, and has dimension equal to the total length of the set of measured marginals. A simple proximal algorithm is given in \cite{mckenna2019graphical} that solves this optimization problem using only repeated calls to a routine to perform marginal inference in a discrete graphical model --- i.e., computing $ M_C(P_{\theta})$ for all $C \in \mathcal{C}$ and various different $\theta$. The procedure is efficient whenever marginal inference in the graphical model is efficient.  Belief propagation is the standard way to perform marginal inference in practice, as it efficiently computes $M_C(P_{\theta})$ directly in terms of $\theta$ without ever explicitly materializing the full joint distribution $P_{\theta}$ \cite{koller2009probabilistic}.  

\begin{remark}[Scalability of \pgm]
\pgm is able to scale to very high-dimensional domains.  The main factors that influence it's scalability are (1) the total size of the parameter vector $\theta$ and (2) the structure of the set $\mathcal{C}$.  The size of $\theta$ is the same as the size of all of the relevant marginals combined, and that must not be too large.  The size of each marginal depends directly on $|\Omega_i|$, the number of possible values for each attribute.  Furthermore, the set $\mathcal{C}$ is important because it corresponds to the structure of the graphical model, and belief propagation is most efficient for tree-structured models.  The scalability of belief propagation and \pgm for non tree-structured models depends on a quantity known as the tree width, which is a measure of how ``tree-like'' the model is \cite{koller2009probabilistic}.
\end{remark}

\begin{remark}[Lack of modeling assumptions]
It is easy to misconstrue the meaning of the graphical model representation. The inferred distribution $P_\theta$ will satisfy conditional independence properties dictated by the structure of the model. However, \emph{no approximation is made when solving the optimization problem in \cref{eq:op}}, and \emph{the independence properties do not arise from assumptions made by the modeler about the structure of the data distribution}. By \cref{fact:op}, there is an optimum to \cref{eq:op} that is a graphical model, and hence satisfies these conditional independence properties. Moreover, the graphical model solution $P_\theta$ can be shown to have maximum entropy among all optima of \cref{eq:op} \cite{mckenna2019graphical}.
\end{remark}

Once $P_{\theta}$ is estimated, \pgm can be used for multiple purposes: reducing error on measured marginals, estimating unmeasured marginals, and even generating synthetic tabular data. These use cases are explained in greater detail below:

\paragraph*{\textbf{Reducing error on measured marginals}}

First, \pgm improves utility by combining all sources of measured information into a single cohesive estimate for the data distribution.  When the measurements are inconsistent with each other, \pgm resolves these inconsistencies in a principled manner, reducing variance and boosting utility.  One achieves this by using $\bar{y}_C = Q_C M_C(P_{\theta})$ in place of the noisy observation $y_C$.  The estimated marginal $\bar{y}_C$ will typically have smaller variance than $y_C$ and will often have lower overall error as well, therefore offering immediate utility improvements at no cost to privacy. \cref{example:boost} demonstrates this idea more concretely in a toy setting.

\begin{example}[Boosting utility on measured marginals] \label{example:boost}
We draw $1000$ tuples from the actual contest dataset and measure two of their marginals, \marginal{SEX,LABFORCE} and \marginal{LABFORCE,SCHOOL}, using the Gaussian mechanism with $\sigma = 50$.  Tables (a--c) below show the true marginals, the noisy marginals, and the marginals estimated by $\pgm$.  One can easily verify that the noisy marginals are not consistent: the \marginal{SEX,LABFORCE} marginal implies the total number of people with \attr{LABFORCE=N} is $ 124.549+318.029=442.578$, while the \marginal{LABFORCE,SCHOOL} marginal implies the same that number is $287.215+171.134=458.349$.  These are two different estimates for the same quantity, which is a consistency problem.  In contrast, the \pgm estimated marginals are consistent: they both agree that the total number is $436.873$.  Additionally, \pgm better estimates the true marginals than the noisy marginals do: the $L_1$ distances are about $213$ and $272$ for \pgm, while they are about $251$ and $295$ for the noisy marginals, which is a significant boost in utility.
\begin{table*}[h!]
\hspace{-2em}
\captionsetup[subtable]{position = below}
\captionsetup[table]{position=top}
\begin{adjustbox}{width=0.91\textwidth}
\begin{subtable}{0.3\textwidth}
\begin{tabular}{cc|c}
\hline
\attr{SEX} & \attr{LABFORCE} & \multicolumn{1}{c}{\emph{count}} \\\hline
M & --- & $156$ \\
M & N & $65$ \\
M & Y & $316$ \\
F & --- & $158$ \\
F & N & $282$ \\
F & Y & $23$ \\\hline\hline
\attr{LABFORCE} & \attr{SCHOOL} & \multicolumn{1}{c}{\emph{count}} \\\hline
--- & N & $159$ \\
--- & Y & $155$ \\
N & N & $288$ \\
N & Y & $59$ \\
Y & N & $336$ \\
Y & Y & $3$ \\\hline
\end{tabular}
\caption{True marginals}
\end{subtable}
\hspace{3em}
\begin{subtable}{0.3\textwidth}
\begin{tabular}{cc|c}
\hline
\attr{SEX} & \attr{LABFORCE} & \multicolumn{1}{c}{\emph{count}} \\\hline
M & --- & $132.428$ \\
\rowcolor{blue!10} M & N & $124.549$ \\
M & Y & $244.365$ \\
F & --- & $173.633$ \\
\rowcolor{blue!10} F & N & $318.029$ \\
F & Y & $-21.358$ \\\hline\hline
\attr{LABFORCE} & \attr{SCHOOL} & \multicolumn{1}{c}{\emph{count}} \\\hline
--- & N & $116.021$ \\
--- & Y & $186.826$ \\
\rowcolor{blue!10} N & N & $287.215$ \\
\rowcolor{blue!10} N & Y & $171.134$ \\
Y & N & $278.498$ \\
Y & Y & $-46.497$ \\\hline
\end{tabular}
\caption{Noisy marginals}
\end{subtable}
\hspace{4em}
\begin{subtable}{0.3\textwidth}
\begin{tabular}{cc|c}
\hline
\attr{SEX} & \attr{LABFORCE} & \multicolumn{1}{c}{\emph{count}} \\\hline
M & --- & $124.829$ \\
\rowcolor{blue!10} M & N & $121.696$ \\
M & Y & $254.636$ \\
F & --- & $166.034$ \\
\rowcolor{blue!10} F & N & $315.177$ \\
F & Y & $0$ \\\hline\hline
\attr{LABFORCE} & \attr{SCHOOL} & \multicolumn{1}{c}{\emph{count}} \\\hline
--- & N & $110.029$ \\
--- & Y & $180.834$ \\
\rowcolor{blue!10} N & N & $276.477$ \\
\rowcolor{blue!10} N & Y & $160.396$ \\
Y & N & $254.636$ \\
Y & Y & $0$ \\\hline
\end{tabular}
\caption{\pgm marginals}
\end{subtable}
\end{adjustbox}
\end{table*}
\end{example}
\vspace{-1em}
\paragraph*{\textbf{Estimating unmeasured marginals}}

Second, \pgm can be used to answer new queries that were never measured directly by using $P_{\theta}$ in place of the true data $D$.  This allows us to estimate new marginals without spending the privacy budget, saving a precious resource. \cref{example:estimating} demonstrates this idea in a toy setting.

\begin{example}[Estimating new marginals] \label{example:estimating}
Building on \cref{example:boost}, recall that we measured the marginals on \marginal{SEX,LABFORCE} and \marginal{LABFORCE,SCHOOL}. We can use \pgm to estimate the marginal on \marginal{SEX,SCHOOL}, even though we never measured it and it can not be directly inferred from the other marginals that were measured.  As shown below in Table (b), the provided estimate is reasonable, given that we never measured it, and we added significant noise to the marginals we did measure.  We reiterate that we obtained this estimate ``for free'', without spending additional privacy budget.  Additionally, \pgm can estimate the marginal on \marginal{SEX,LABFORCE,SCHOOL}, which is shown in Table (d).  This is pretty close to the true 3-way marginal, shown in Table (c).  In fact, the normalized $L_1$ error is only $0.135$.  While there may be other equally good estimates for the 3-way marginal (according to the loss function in \cref{eq:op}), the estimate provided by \pgm has maximum entropy among all of them.  In this case, \pgm was fairly accurate because \attr{SEX} and \attr{SCHOOL} are (approximately) conditionally independent given \attr{LABFORCE} in the true data.
\begin{table*}[h]
\hspace{-2em}
\captionsetup[subtable]{position = below}
\captionsetup[table]{position=top}
\begin{adjustbox}{width=0.91\textwidth}
\begin{subtable}{0.3\linewidth}
\begin{tabular}{cc|r}
\hline
\attr{SEX} & \attr{SCHOOL} & \multicolumn{1}{c}{\emph{count}} \\\hline
M & N & $\quad423$ \\
M & Y & $114$ \\
F & N & $360$ \\
F & Y & $103$ \\
\end{tabular}
\caption{True 2-way marginal\quad\quad\quad}
\vspace{2em}
\begin{tabular}{cc|r}
\hline
\attr{SEX} & \attr{SCHOOL} & \multicolumn{1}{c}{\emph{count}} \\\hline
M & N & $378.873$ \\
M & Y & $122.289$ \\
F & N & $262.269$ \\
F & Y & $218.942$ \\
\end{tabular}
\caption{Estimated 2-way marginal}
\end{subtable}
\begin{subtable}{0.3\linewidth}
\begin{tabular}{ccc|r}
\hline
\texttt{SEX} & \texttt{LABFORCE} & \attr{SCHOOL} & \multicolumn{1}{c}{\emph{count}} \\\hline
M & --- & N & $74$ \\
M & --- & Y & $82$ \\
M & N   & N & $36$ \\
M & N   & Y & $29$ \\
M & Y   & N & $313$ \\
M & Y   & Y & $3$ \\
F & --- & N & $85$ \\
F & --- & Y & $73$ \\
F & N   & N & $252$ \\
F & N   & Y & $30$ \\
F & Y   & N & $23$ \\
F & Y   & Y & $0$ \\
\multicolumn{4}{c}{}
\end{tabular}
\vspace{-1em}
\caption{True 3-way marginal}
\end{subtable}
\hspace{4em}
\begin{subtable}{0.3\linewidth}
\begin{tabular}{ccc|r}
\hline
\texttt{SEX} & \texttt{LABFORCE} & \attr{SCHOOL} & \multicolumn{1}{c}{\emph{count}} \\\hline
M & --- & N & $47.221$ \\
M & --- & Y & $77.608$ \\
M & N   & N & $77.016$ \\
M & N   & Y & $44.68$ \\
M & Y   & N & $254.636$ \\
M & Y   & Y & $0.000$ \\
F & --- & N & $62.808$ \\
F & --- & Y & $103.226$ \\
F & N   & N & $199.461$ \\
F & N   & Y & $115.716$ \\
F & Y   & N & $0.000$ \\
F & Y   & Y & $0.000$ \\
\multicolumn{4}{c}{}
\end{tabular}
\vspace{-1em}
\caption{Estimated 3-way marginal}
\end{subtable}
\end{adjustbox}
\end{table*}
\end{example}
\vspace{-1.5em}
\paragraph*{\textbf{Generating synthetic data}}

Third, \pgm can be used to generate synthetic data $\bar{D}$ in tabular format.  $\bar{D}$ can be used in place of $P_{\theta}$ and will generally give similar results.  They will not give exactly the same results because $\bar{D}$ has integer-valued marginals while $P_{\theta}$ has real-valued marginals, so some additional rounding error is unavoidable.  One can obtain the synthetic data in multiple ways; a simple and natural approach would be to sample records from $P_{\theta}$ to form a synthetic dataset.  This naive approach would introduce sampling error which is undesirable.  
\pgm uses an alternative approach to reduce error from additional sources of randomness. The details of this procedure are available in the open-source implementation of \pgm, and are summarized in the supplementary material.  In \cref{example:sample}, we give an intuitive idea of how this procedure works, and illustrate why it is preferable to the sampling approach. 
 
\begin{example}[Generating synthetic data] \label{example:sample}
Building on \cref{example:boost,example:estimating}, we calculate the \attr{LABFORCE} marginal from the \pgm model in Table (a) below, which has fractional counts.   We also use \pgm to generate synthetic data and show the same marginal in Table (b), which has integer counts.  These two marginals almost exactly match, because \pgm tries to preserve the model marginals as closely as possible when generating synthetic data.  However, synthetic data obtained by i.i.d sampling will not match the model marginals as closely due to the randomness in sampling, as shown in Table (c).  
\begin{table*}[h]
\hspace{-2em}
\captionsetup[subtable]{position = below}
\captionsetup[table]{position=top}
\begin{adjustbox}{width=0.91\textwidth}
\begin{subtable}{0.3\linewidth}
\begin{tabular}{c|c}
\hline
\attr{LABFORCE} & \multicolumn{1}{c}{\emph{count}} \\\hline
--- & $290.863$ \\
N & $436.873$ \\
Y & $254.636$ \\
\end{tabular}
\caption{\pgm \\ model marginal}
\end{subtable}
\begin{subtable}{0.3\linewidth}
\begin{tabular}{c|c}
\hline
\attr{LABFORCE} & \multicolumn{1}{c}{\emph{count}} \\\hline
--- & $291$ \\
N & $437$ \\
Y & $254$ \\
\end{tabular}
\caption{\pgm \\ synthetic data marginal}
\end{subtable}
\begin{subtable}{0.3\linewidth}
\begin{tabular}{c|c}
\hline
\attr{LABFORCE} & \multicolumn{1}{c}{\emph{count}} \\\hline
--- & $262$ \\
N & $468$ \\
Y & $252$ \\
\end{tabular}
\caption{Sampled \\ synthetic data marginal}
\end{subtable}
\end{adjustbox}
\end{table*}
\end{example}

\section{Algorithm Description} \label{sec:mech}

\begin{algorithm}[t]
\algsize{
\begin{tabularx}{\textwidth}{llr}
\nl \textbf{Calibrate Noise} & Derive noise scale $\sigma$ from target privacy parameters $(\epsilon, \delta)$ & \cref{eq:calibrate} \\
\nl \textbf{Encode Domain} &  Use public information to encode attribute domains &  \\
\nl \textbf{Transform Data} & Transform data using insights from provisional data & \cref{alg:preprocess} \\
\nl \textbf{Compress Domain} & Use data to reduce domain: & \\
\quad \textbf{Measure} & \quad Measure all one-way marginals & \cref{alg:measure} \\
\quad \textbf{Compress} & \quad Remove domain elements failing threshold test & \cref{alg:compress} \\
\nl \textbf{Select Marginals} & Select a subset of 2- and 3-way marginals & \cref{alg:select} \\
\nl \textbf{Measure Marginals} & Measure selected marginals  & \cref{alg:measure} \\
\nl \textbf{Synthesize data} & Synthesize records using \pgm:  &  \\
\quad \textbf{Estimate} & \quad Estimate distribution from Step 4 and 6 measurements\xspace\xspace\xspace & \cref{eq:op2} \\
\quad \textbf{Generate} & \quad Generate synthetic records  & \cref{alg:synthdata} \\
\nl \textbf{Reverse} & Reverse the transformation made in Step 3 & \cref{alg:reverse} \\
\end{tabularx}
\caption{\label{alg:mech}\mech}}
\end{algorithm}

In this section we describe \mech, which takes the basic mechanism template outlined in the previous section, and applies it to the setting of the NIST competition.  
\mech simply invokes the Gaussian mechanism to measure a carefully chosen subset of 1, 2, and 3-way marginals.  Then the resulting noisy measurements are post-processed using \pgm to obtain synthetic data that is most consistent with those marginals.  
\mech does not follow the template from the previous section exactly, as there are two rounds of measurements, and an additional domain compression step developed specifically to deal with some of the challenges around the dataset used in the competition.  The high-level steps of \mech are stated in \cref{alg:mech}, and a detailed description of each step will be provided in this section, with motivations and intuitions for the various design choices.

\paragraph*{\textbf{Step 1: Calibrate Noise}}

In this step, $\sigma$ is calibrated to ensure the whole algorithm satisfies $(\epsilon, \delta)$-differential privacy.  Note that only steps (4) and (6) in \cref{alg:mech} use the sensitive data, and these are both invocations of \cref{alg:measure}, which is  $ (\alpha, \frac{\alpha}{2 \sigma^2})$-RDP (\cref{thm:privacymeasure}).  The data transformations made in steps (3) and (4) do not affect the privacy analysis of \cref{alg:measure} since one individual can still only affect each marginal by at most one.  Hence \mech is $(\alpha, \frac{\alpha}{\sigma^2})$-RDP by two-fold adaptive composition (\cref{prop:composition}).  Moreover, by \cref{prop:rdpdp}, \mech is $\big(\frac{\alpha}{\sigma^2} + \frac{\log{(1/\delta)}}{\alpha-1}, \delta\big)$-DP for all $\alpha \geq 1$.  

For a fixed $\alpha$, it is easy to determine $\sigma$ by solving the equation $ \frac{\alpha}{\sigma^2} + \frac{\log{(1 / \delta)}}{\alpha-1} = \epsilon$ for $\sigma$.   The best value of $\sigma$ can be obtained by minimizing over all $\alpha$.  In the contest, this computation was done by invoking the \texttt{moments accountant} \cite{abadi2016deep}, which minimizes over $\alpha = 1, \dots, 512$.  However, this minimization can actually be done in closed form \cite{zhao2019reviewing}, leading to  the following equation for $\sigma$:
\begin{equation} \label{eq:calibrate}
\sigma = \frac{\sqrt{\log{(1 / \delta)}} + \sqrt{\log{(1 / \delta)} + \epsilon}}{\epsilon}
\end{equation}

\begin{theorem}[Privacy of \mech] \label{thm:privacy}
\cref{alg:mech} is $(\epsilon, \delta)$-differentially private.
\end{theorem}

\begin{proof}
From the analysis above, we know that \mech is $\Big(\frac{\alpha}{\sigma^2} + \frac{\log{(1/\delta)}}{\alpha-1}, \delta\Big)$-DP for all $\alpha \geq 1$.  By plugging in $ \alpha = 1 + \sigma \sqrt{\log{(1/\delta)}}$ and $\sigma = \frac{\sqrt{\log{(1 / \delta)}} + \sqrt{\log{(1 / \delta)} + \epsilon}}{\epsilon}$ and simplifying, we see that $\frac{\alpha}{\sigma^2} + \frac{\log{(1/\delta)}}{\alpha-1} = \epsilon$, and hence \mech is $(\epsilon, \delta)$-DP as desired.   The algebraic manipulation is routine but messy; it can easily be verified with \texttt{sympy} (see \cref{fig:proof-privacy} in the supplement).
\end{proof}

\begin{remark}[Noise Calibration]
It is well known that calibrating $\sigma$ via an RDP analysis does not give the smallest possible value required to achieve $(\epsilon, \delta)$-DP, and an analytic calibration gives strictly better results \cite{balle2018improving}, at least for a \emph{single} invocation of the Gaussian mechanism.  However, at the time of the competition, adaptive composition of \emph{two} Gaussian mechanisms was needed, and RDP was chosen because of its clean and well-understood guarantees.  If using the analytic Gaussian mechanism, advanced composition would be necessary to reason about the privacy of two-fold adaptive composition \cite{Dwork14Algorithmic}.  Since these are somewhat loose bounds, the benefit of the analytic calibration would be lost.  However, since the time of the competition, much progress has been made on understanding the behavior of the Gaussian mechanism under composition and it is now known that the analytic Gaussian mechanism can be used to calibrate noise for multiple (adaptive) invocations of the Gaussian mechanism \cite{dong2019gaussian,sommer2019privacy}.  This would give a smaller value of $\sigma$ than the one shown in \cref{eq:calibrate}, typically offering an improvement of 10 to 20 percent.    
\end{remark}

\paragraph*{\textbf{Step 2: Encode Domain}}

\begin{table}[t!]
\resizebox{\columnwidth}{!}{%
\begin{tabular}{ll|ll|ll|ll|ll}
\midrule
1. \attr{SPLIT} & 2/2 & 21. \attr{NCHLT5} & 7/7 & 41. \attr{SIZEPL} & 31/19 & 61. \attr{SUPDIST} & 631/631 & 81. \attr{ENUMDIST} & 3021/3021 \\
2. \attr{SLREC} & 3/2 & 22. \attr{RACE} & 7/7 & 42. \attr{EMPSTATD} & 35/15 & 62. \attr{METAREA} & 657/334 & 82. \attr{CITYPOP} & 3225/3225 \\
3. \attr{SEX} & 3/2 & 23. \attr{WKSWORK2} & 7/7 & 43. \attr{WKSWORK1} & 53/53 & 63. \attr{PRESGL} & 816/816 & 83. \attr{URBPOP} & 3225/3225 \\
4. \attr{SCHOOL} & 3/2 & 24. \attr{VET1940} & 9/4 & 44. \attr{SEA} & 54/54 & 64. \attr{BPL} & 901/163 & 84. \attr{METAREAD} & 6561/378 \\
5. \attr{URBAN} & 3/3 & 25. \attr{UCLASSWK} & 9/8 & 45. \attr{OCCSCORE} & 81/81 & 65. \attr{MBPL} & 901/164 & 85. \attr{MTONGUED} & 9602/489 \\
6. \attr{FARM} & 3/3 & 26. \attr{VETPER} & 9/8 & 46. \attr{AGEMARR} & 90/89 & 66. \attr{FBPL} & 901/165 & 86. \attr{MIGMET5} & 10000/379 \\
7. \attr{OWNERSHP} & 3/3 & 27. \attr{HISPRULE} & 9/9 & 47. \attr{MIGRATE5D} & 91/15 & 67. \attr{IND1950} & 998/162 & 87. \attr{MIGCOUNTY} & 10000/385 \\
8. \attr{RESPONDT} & 3/3 & 28. \attr{HRSWORK2} & 9/9 & 48. \attr{MTONGUE} & 97/92 & 68. \attr{MIGSEA5} & 998/510 & 88. \attr{ERSCOR50} & 10000/1002 \\
9. \attr{SPANNAME} & 3/3 & 29. \attr{CLASSWKR} & 10/4 & 49. \attr{SEI} & 97/97 & 69. \attr{OCC} & 999/231 & 89. \attr{EDSCOR50} & 10000/1002 \\
10. \attr{LABFORCE} & 3/3 & 30. \attr{INCNONWG} & 10/4 & 50. \attr{CLASSWKRD} & 99/18 & 70. \attr{EDUCD} & 1000/44 & 90. \attr{NPBOSS50} & 10000/1002 \\
11. \attr{VETWWI} & 3/3 & 31. \attr{SAMEPLAC} & 10/4 & 51. \attr{HRSWORK1} & 99/99 & 71. \attr{HIGRADED} & 1000/69 & 91. \attr{CITY} & 10000/1164 \\
12. \attr{SSENROLL} & 3/3 & 32. \attr{VETSTAT} & 10/4 & 52. \attr{VETSTATD} & 100/10 & 72. \attr{GQTYPED} & 1000/92 & 92. \attr{MIGCITY5} & 10000/1164 \\
13. \attr{METRO} & 4/4 & 33. \attr{VETCHILD} & 10/5 & 53. \attr{GQFUNDS} & 100/13 & 73. \attr{IND} & 1000/136 & 93. \attr{RENT} & 10000/10000 \\
14. \attr{EMPSTAT} & 4/4 & 34. \attr{MIGTYPE5} & 10/6 & 54. \attr{EDUC} & 100/13 & 74. \attr{UIND} & 1000/136 & 94. \attr{MBPLD} & 90021/537 \\
15. \attr{HISPAN} & 5/5 & 35. \attr{SAMESEA5} & 10/6 & 55. \attr{AGEMONTH} & 100/15 & 75. \attr{MIGPLAC5} & 1000/199 & 95. \attr{FBPLD} & 90021/539 \\
16. \attr{CITIZEN} & 5/5 & 36. \attr{MIGRATE5} & 10/7 & 56. \attr{HIGRADE} & 100/25 & 76. \attr{UOCC} & 1000/231 & 96. \attr{BPLD} & 90022/536 \\
17. \attr{WARD} & 6/6 & 37. \attr{GQTYPE} & 10/10 & 57. \attr{CHBORN} & 100/62 & 77. \attr{UOCC95} & 1000/279 & 97. \attr{VALUEH} & 10000000/5003 \\
18. \attr{NATIVITY} & 6/6 & 38. \attr{MARRNO} & 10/10 & 58. \attr{AGE} & 109/109 & 78. \attr{OCC1950} & 1000/283 & 98. \INCWAGEA & ---/52 \\
19. \attr{MARST} & 7/6 & 39. \attr{OWNERSHPD} & 21/8 & 59. \attr{HISPAND} & 481/55 & 79. \attr{DURUNEMP} & 1000/1000 & 99. \INCWAGEB & ---/8 \\
20. \attr{GQ} & 7/7 & 40. \attr{FAMSIZE} & 22/22 & 60. \attr{RACED} & 621/238 & 80. \attr{COUNTY} & 1251/385 & \quad \:\: \INCWAGE & 10000000/--- \\
\bottomrule
\end{tabular}%
}
\caption{ \label{table:domain} Domain information for the census dataset used in the third round of the competition.  Table specifies attribute names, and number of possible values for that attribute according to (1) the provided specs file and (2) the specs file combined with IPUMS documentation.}
\end{table}

Before running \mech, it is necessary to know the data domain $\Omega$.  The supplied competitor pack came with a ``SPECS'' file that contained some domain information.  Specifically, it supplied a single positive integer $n_i$ for each attribute $i$, and the domain for that attribute was assumed to be $\dom^{\text{SPECS}}_i = \set{0, \dots, n_i - 1}$.  However, because the data is derived from a census source, the domain is very thoroughly documented on the Integrated Public Use Microdata Series (IPUMS) website \cite{ipums}.  IPUMS offers a much finer grained view of the data domain, specifying the exact set of possible values for most attributes.  We use $\dom^{\text{IPUMS}}_i$  to denote the domain of possible values for attribute $i$ according to IPUMS.  \cref{example:ipums} demonstrates the benefit of using the finer grained IPUMS domain information. 

\begin{example}[SPECS vs. IPUMS] \label{example:ipums}
From the specs file and the IPUMS website, we see the domain for the \attr{EDUC} attribute is $\domsp_i = \set{0, 1, \dots, 99}$ and $\domip_i = \set{0, 1, \dots, 11, 99}$.  Note that $99$ is a special code that typically corresponds to missing data.  While both sources agree that $99$ is the largest possible value, the IPUMS documentation suggests that values in the range $12, \dots, 98$ are not possible.  Using the finer granularity domain from IPUMS reduces the number of possible values for \attr{EDUC} from $100$ to $13$.  This has two important ramifications.  First, it will make \pgm more efficient in later steps, since the scalability of that tool depends directly on the domain sizes of the attributes.  Second, it will prevent \mech from inadvertently introducing out-of-domain tuples to the synthetic data which could otherwise occur by adding positive noise to zero counts. 
\end{example}

Often $\domip_i$ is a subset of $\domsp_i$, although this is not always the case.  In some cases, IPUMS documents a certain value as being possible that never appeared in the provisional dataset or the supplied specs file.  To account for this \mech uses the intersection of the two domains, i.e., $\dom_i = \domsp_i \cap \domip_i$.  \cref{table:domain} enumerates the attributes in the dataset along with the domain size provided in the specs file, and the compressed domain size derived by \mech.

\paragraph*{\textbf{Step 3: Transform Data}}

\begin{algorithm}[t]
\algsize{
\textbf{Input:} $D$ (sensitive dataset) \\
\textbf{Output:} $D$ (transformed sensitive dataset) \\
\nl Replace \attr{VALUEH} attribute in $D$ using transformation: \\
\nl Split \INCWAGE attribute in $D$ into two attributes \INCWAGEA and \INCWAGEB using transformation:\footnotemark \\
\vspace{-1.1em}
\hspace{-3em}
\scriptsize{
\begin{minipage}{0.5\textwidth}
\begin{align*} 
&\attr{VALUEH} = \\
&\quad \begin{cases}
\floor{\frac{\attr{VALUEH}}{5}} & \attr{VALUEH} \leq \num{25000} \\
5001 & \attr{VALUEH} = \num{9999998} \\
5002 & \attr{VALUEH} = \num{9999999} \\
5000 & \text{otherwise}
\end{cases} \\
&\INCWAGEA =  \\
&\quad \begin{cases}
\floor{\frac{\INCWAGE}{100}} & \INCWAGE \leq 5000 \\
50 & \INCWAGE < \num{9999998} \\
51 & \INCWAGE = \num{9999998} \\
\end{cases} \end{align*}
\end{minipage}%
\begin{minipage}{0.5\textwidth}
\begin{align*} &\INCWAGEB = \\ & \quad\begin{cases}
0 & \INCWAGE \equiv 0 \mod 100 \\
1 & \INCWAGE \equiv 0 \mod 20 \\
2 & \INCWAGE \equiv 0 \mod 50 \\
3 & \INCWAGE \equiv 0 \mod 25 \\
4 & \INCWAGE \equiv 0 \mod 10 \\
5 & \INCWAGE \equiv 0 \mod 5 \\
6 & \INCWAGE \equiv 0 \mod 2 \\
7 & \text{otherwise} \\
\end{cases} \end{align*}
\end{minipage}}}
\caption{\label{alg:preprocess}Transform data}
\end{algorithm}
\footnotetext{Each condition in the piecewise definition of \INCWAGEA and \INCWAGEB should be interpreted as an ``else if'' statement rather than an ``if'' statement, as clearly multiple conditions can be true at the same time}

In addition to the general domain encoding outlined above, \mech gave special attention to two of the attributes with the largest domain: \attr{INCWAGE} and \attr{VALUEH}.  Both of these attributes started out with $10$ million possible values, and the IPUMS documentation provided limited information on these attributes.  Therefore, \mech leveraged the provisional dataset to try to identify a domain that captured all or most of the observed values for these attributes.  \cref{alg:preprocess} shows how these attributes are transformed to reduce the domain size.  
The intuition behind this pre-processing procedure is to compress the domain, while ensuring the compressed domain still covers all or most of the values observed in the provisional dataset.  For example, other than the special codes of \num{9999998} and \num{9999999}, $99.2\%$ of records have \attr{VALUEH} that is a multiple of $5$ and less than or equal to $\num{25000}$.\footnote{\num{9999998} and \num{9999999} are special codes corresponding to missing values of N/A values.}  This allows us to compress the domain of \attr{VALUEH} to a much more manageable size of $5003$ while still covering about $99.7\%$ of the observed values.  

\mech uses a similar approach to handle \attr{INCWAGE}.  Over $99.95\%$ of records in the provisional dataset had an \attr{INCWAGE} value of either \num{9999998} or something in the range $[0,5000]$.  For that reason, it is reasonably safe to truncate values above $5000$ without introducing too much bias. 
This transformation reduces the domain size of \INCWAGE to $5002$, but \mech takes things one step further.  
There are clear periodic patterns in the \INCWAGE marginal, as the most common values are all multiples of $100$.  Multiples of $20$, $50$, and $25$ are also common.  To exploit this observation, \mech splits up \attr{INCWAGE} into two attributes: \INCWAGEA and \INCWAGEB, and never measures \INCWAGE directly, but only indirectly through these two derived attributes.  
\INCWAGEA is meant to capture the coarse-grained income by discretizing it into width $100$ bins, whereas \INCWAGEB is meant to capture the periodicity in the last two digits.  These two derived attributes have smaller domains of size $52$ and $8$, respectively. The exact formulas are given in \cref{alg:preprocess}.



\paragraph*{\textbf{Step 4: Compress Domain}}

\begin{algorithm}[t]
\algsize{
\textbf{Input:} log (list of noisy measurements), $D$ (sensitive dataset), $\Omega$ (domain) \\
\textbf{Output:} $D$ (transformed sensitive dataset), $\Omega$ (transformed domain) \\
\nl \For{\emph{For each measurement $(\_\_, \tilde{\mu}, \sigma, \set{i})$ in log}}{
\nl Replace values for attribute $i$ in dataset $D$ using transformation:
$$ t \leftarrow \begin{cases} t & \tilde{\mu}_t \geq 3 \sigma \\
\varnothing & \text{otherwise} \end{cases} $$
\nl Modify domain accordingly, $\Omega_i \leftarrow \set{ t \mid \tilde{\mu}_t \geq 3 \sigma } \cup \set{\varnothing}$
}}
\caption{\label{alg:compress}Domain compression}
\end{algorithm}

In step (1), \mech was able to greatly reduce the domain size by incorporating information from IPUMS.  However, even after this domain encoding, some of the attributes in the data remain fairly sparse.  For example, only $17.4\%$ percent of counts in the \attr{VALUEH} marginal exceed $100$.  In this step, we answer all 1-way marginals, i.e., we pass $\mathcal{C} = \set{ \set{i} \mid i \in \mathcal{A} }$ into \cref{alg:measure}.  Every marginal is assigned an equal weight of $1$, with the exception of \INCWAGEA, which is assigned a weight of $2$.  

After obtaining noisy 1-way marginals, \cref{alg:compress} is called, which searches for domain elements for which the noisy count fell below the threshold of $3 \sigma$.  These domain elements were merged into a single ``other'' domain element, denoted $\varnothing$.  Later steps of \mech operate over the resulting transformed dataset and domain.  This step has two main benefits, similar to the ones from the domain encoding step.  First, it improves scalability of \pgm in later steps by reducing the domain size of the attributes.  Second, it ensures that the tuples generated by \mech (probably) have attribute values that actually occurred in the dataset.

\paragraph*{\textbf{Step 5: Select Measurements}}

The next step of \mech is to identify a collection of 2- and 3-way marginals that will later be measured.  This is one of the most crucial components of \mech, because the marginals selected in this step will ultimately determine the marginals that will be preserved in the generated synthetic data.  
It is important to note that this step selects measurements without using the sensitive dataset, although it does rely heavily on the provisional dataset. This algorithm does take the privacy budget $\epsilon$ as input, but it does not ``consume'' it --- it only uses it to determine weights to assign to each selected marginal.

\begin{algorithm}[t]
\algsize{
\textbf{Input:} $\hat{D}$ (the provisional dataset), $\epsilon$ (privacy budget) \\
\textbf{Output:} $\mathcal{C}$ (a collection of attribute subsets), $w$ (weights for each $C \in \mathcal{C})$ \\
\nl Construct a complete graph $G$ where vertices are attributes in the dataset, and edge $(i,j)$ is weighted according to the mutual information between attribute $i$ and attribute $j$ in the dataset $\hat{D}$.  Add $100$ to the edge weights for \marginal{SEX,CITY}, \marginal{SEX,\INCWAGEA} and \marginal{CITY,\INCWAGEA}.   \\
\nl Identify the maximum spanning tree (MST) of the graph, and for each edge in the tree, add the attribute pair to $\mathcal{C}$.  Also add \marginal{SEX,CITY}, \marginal{SEX,\INCWAGEA}, \marginal{CITY,\INCWAGEA}, and \marginal{SEX,CITY,\INCWAGEA} to $\mathcal{C}$ if they are not already included. \\
\nl For each pair of adjacent edges $(i,j), (i,k)$ in the MST, compute the marginals $M_{ij}(\hat{D})$, $M_{ik}(\hat{D})$, and $M_{ijk}(\hat{D})$.  Use \pgm to estimate $\tilde{M}_{ijk}$ from $M_{ij}$ and $M_{ik}$, and record the error in the estimate $E_{ijk} = \norm{M_{ijk} - \tilde{M}_{ijk}}_1$.  \\
\nl \revision{For each attribute $i$, construct a complete graph consisting of nodes that are neighbors of $i$ in the MST, and each edge $(j,k)$ in the new graph is assigned a weight of $E_{ijk}$.  Remove edges whose weight is below a threshold of $0.1$, and compute the maximum spanning tree of the resulting graph.  For each edge $(j,k)$ in the new MST, add the $(j,k)$ and $(i,j,k)$ marginals to $\mathcal{C}$.}\\
\nl Remove attribute subsets whose marginal is too large, i.e., $C \in \mathcal{C}$ such that $\prod_{i \in C} n_i \geq 10^6$. \\
\nl Assign weights to selected attribute subsets using formula:}
\scriptsize{
$$ w_C \propto \begin{cases}
8 & C = \marginal{SEX,CITY,\INCWAGEA}, \epsilon \leq 0.3 \\
4 & C = \marginal{SEX,CITY,\INCWAGEA}, \epsilon \geq 4.0 \\
6 & C = \marginal{SEX,CITY,\INCWAGEA}, 0.3 < \epsilon < 4.0 \\
2 & C \in \set{\marginal{SEX,CITY},\marginal{SEX,\INCWAGEA},\marginal{CITY,\INCWAGEA}} \\
1 & \text{otherwise}
\end{cases} $$}
\caption{\label{alg:select}Marginal selection algorithm}
\end{algorithm}

\cref{alg:select} shows how \mech selects 2- and 3-way marginals for measurement.  This algorithm is inspired by a similar approach used by two other mechanisms for differentially private synthetic data \cite{zhang2017privbayes,chen2015differentially}.  It combines one principled step, which is to find the maximum spanning tree (MST) on the graph where edge weights correspond to mutual information between two attributes, with some additional heuristics to ensure that certain important attribute pairs are selected, and more heuristics to select some triples while keeping the graph tree-like.  A reader familiar with graphical models with recognize the MST step as the famous Chow-Liu algorithm for structure learning in a graphical model \cite{chow1968approximating}.  Intuitively, highly correlated marginals should be measured because attributes that are independent can trivially be preserved without direct measurement.  

\cref{fig:queries} shows the marginals selected by this algorithm in graphical format.  Each edge in the graph represents a pair of attributes whose marginal will be measured by \mech, and each triangle in the graph represents a triple of attributes whose marginal will be measured by \mech.  The structure of this measurement graph also corresponds to the structure of the graphical model used by \pgm.  The tree-like structure of the graph will ultimately allow \pgm to run efficiently.  The green subgraph corresponds to the \marginal{SEX,CITY,\INCWAGEA} clique, the black edges form a maximum spanning tree of the underlying correlation graph, and the dotted red edges are additional edges that enhance the expressive capacity of the model while retaining the tree-like structure. 
\begin{figure}
\includegraphics[width=\textwidth]{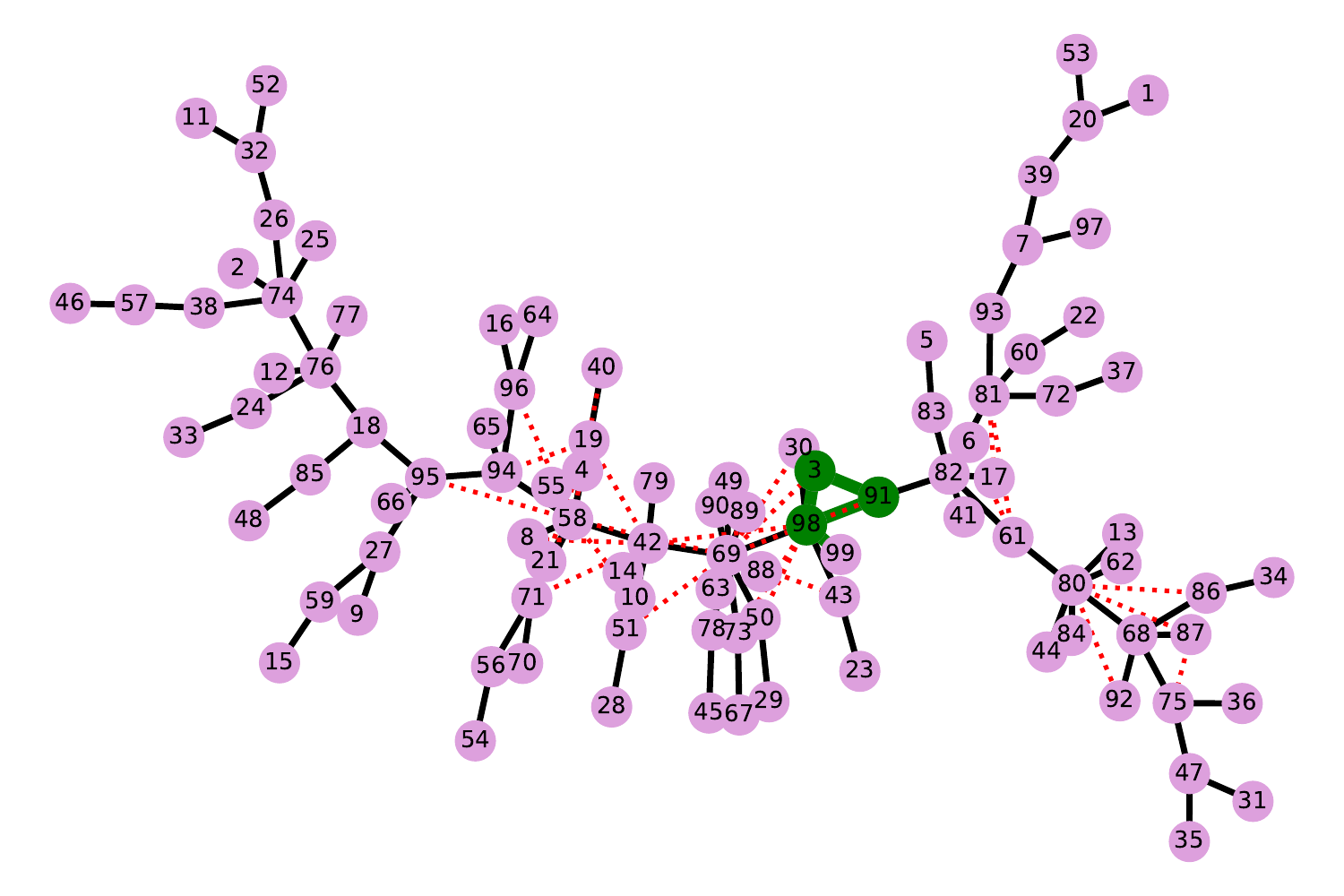}
\vspace{-2em}\caption{\label{fig:queries}A graphical depiction of the 2- and 3-way marginals selected by \mech.  Nodes correspond to attributes in the dataset, and edges correspond to marginals selected by \mech.  Nodes are labeled by the 2-digit code for the attribute given in \cref{table:domain}.  Black edges form a \emph{maximum spanning tree} of the underlying correlation graph.  Green nodes and edges correspond to the special \marginal{SEX, CITY, \INCWAGEA} attributes.  Dotted red edges identify the extra marginals chosen to improve expressive capacity of the model while maintaining tractability of \pgm.  All dotted red edges form a triangle, and for each of those \mech also included the 3-way marginal corresponding to the three nodes that make up the triangle in the list of selected marginals.\vspace{-0.7em}} 
\end{figure}
\paragraph*{\textbf{Step 6: Measure Marginals}}

The next step of \mech is to measure the marginals selected in the previous step with \cref{alg:measure}.  The result is a collection of noisy measurements contained within a measurement log, and suitable for post-processing with \pgm.

\paragraph*{\textbf{Step 7: Synthesize data}}

The next step of \mech is to combine the measurement logs from Steps 4 and 6 and pass them to \pgm, which returns a synthetic dataset whose marginals approximately match those in the measurement log.  Because the measurements in Steps 4 and 6 were made on the uncompressed and compressed domains, respectively, the measurements from Step 4 had to be re-expressed over the compressed domain.

\paragraph*{\textbf{Step 8: Reverse Transformation}}

The final step of \mech is to reverse the transformations made in Steps 4 and 3, to bring the data back to the original domain.  For Step 4, this requires evenly distributing any instances of $\varnothing$ among the original domain elements mapped to it.  For Step 3, this requires modifying \attr{VALUEH} and combining \INCWAGEA and \INCWAGEB back into a single attribute.  The details are given in \cref{alg:reverse} of the appendix.

\section{Extensions} \label{sec:extensions}

\begin{algorithm}[t]
\algsize{
\textbf{Input:} $D$ (sensitive dataset), log (measurements of 1-way marginals), \revision{$\rho$ (privacy parameter)}, $\mathcal{C}$ (initial set of $(i,j)$ pairs to measure; empty by default) \\
\textbf{Output:} $\mathcal{C}$ (final set of $(i,j)$ pairs to measure) \\
\nl Use \pgm to estimate all 2-way marginals $\bar{M}_{ij}$ from log \\
\nl Compute $L_1$ error between estimated 2-way marginal and actual 2-way marginal for all $i,j$: $q_{ij}(D) = \norm{M_{ij}(D) - \bar{M}_{ij}}_1$ (this is a sensitivity 1 quantity) \\
\nl Let $G = (\mathcal{A}, \mathcal{C})$ be the graph where attributes are vertices and edges are pairs of attributes \\
\nl Let $r$ be the number of connected components in $G$ \protect\footnotemark \\
\nl \revision{Let $\epsilon = \sqrt{\frac{8 \rho}{r-1}}$} \\
\nl \For{\emph{Repeat $r - 1$ times}}{
\nl Let $S$ be the set of all attribute pairs $(i,j)$, where $i$ and $j$ are in different connected components of $G$ \\
\nl \revision{Select attribute pair $(i,j) $ by running the \texttt{exponential mechanism} with quality score function $q_{ij}$ on set $S$ and privacy parameter $\epsilon$.} \\
\nl Add attribute pair $(i,j)$ to $\mathcal{C}$ \\
}}
\caption{\label{alg:mst}Differentially private measurement selection}
\end{algorithm}
\footnotetext{If $\mathcal{C}$ is empty, this is just the number of attributes $d$.}

One limitation with \mech is that it is highly tailored to the setting of the NIST competition, and crucially relies on the existence of a public provisional dataset that can be used to select marginals.   In more general settings, we will not always have access to a provisional dataset that follows a similar distribution as the sensitive data.  For that reason, we propose \texttt{MST}, a general purpose mechanism that is inspired by the \mech mechanism, but doesn't rely on the existence of provisional data.  The basic mechanism is the same as \mech outlined in \cref{alg:mech}, with a couple minor exceptions.  First, the preprocessing transformations and corresponding reverse transformations are not done --- those were specific to the U.S. Census dataset used in the competition and not generally applicable beyond that setting.    Second, the measurement selection step, which previously relied on a provisional dataset to select correlated marginals, is replaced by a differentially-private version that uses the sensitive dataset.  \revision{\mst devotes $\frac{1}{3}$ of the RDP budget towards measurement selection, and uses the remaining $\frac{2}{3}$ of the RDP budget for measuring the marginals.   Privacy of $\mst$ follows by adaptive composition \cref{prop:composition}.  For completeness, this calculation is given in the appendix.}

The measurement selection algorithm is shown in \cref{alg:mst}.  Just like \cref{alg:select},  this algorithm tries to find a collection of attribute pairs that form a maximum spanning tree of an underlying correlation graph.  However, as it uses the sensitive dataset, it must do this in a differentially private way.  To achieve this, we first use a low-sensitivity approximation of the mutual information for assigning edge weights.  We assume that we already measured all 1-way marginals, so we can get reasonable estimates of 2-way marginals by invoking \pgm.\footnote{In this simple case, \pgm estimates 2-way marginals under an independence assumption, which could alternatively be achieved by multiplying the (noisy) one-way marginals together.}  The edge weights are then computed as the $L_1$ distance between the true 2-way marginal and the estimated one (a sensitivity $1$ quantity).   After computing the edge weights, \cref{alg:mst} can be seen as a differentially private version of Kruskal's algorithm \cite{kruskal1956shortest} for computing a maximum spanning tree.  It consists of $d-1$ steps (the number of edges in a spanning tree), and in each step, it adds a highly weighted edge that connects two different connected components.  In Kruskal's algorithm, the highest weighted edge is chosen, but this would not be differentially private.  \revision{We instead invoke the \texttt{exponential mechanism} to select a highly weighted edge in a differentially private way.  In principle, we could apply any private selection algorithm here, including report-noisy-max \cite{Dwork14Algorithmic} and the recently developed permute-and-flip mechanism \cite{mckenna2020permute}.  While permute-and-flip is known to dominate the exponential mechanism under $\epsilon$-DP \cite{mckenna2020permute}, the exponential mechanism enjoys a tighter privacy analysis under R\'enyi-DP \cite{cesar2021bounding}.}

The result of this algorithm is a collection of attribute pairs that will be measured by \mst.  \cref{alg:mst} has an optional argument, $\mathcal{C}$, which is an initial set of attribute pairs to measure.  If this is supplied, the algorithm will always include those in the result, and then constructs a maximum spanning tree around them.  This enables some extra flexibility that may be beneficial in certain settings where some marginals are more important than others, and need to be preserved even if they are not the most highly correlated.  For many applications, $\mathcal{C}$ can just be empty.  In the context of the competition, this feature is useful because the marginals relating \attr{SEX}, \attr{CITY}, and \INCWAGEA are very important (since their accuracy determines $\frac{1}{3}$ of the final score).

\revision{
\begin{theorem}[Privacy of \cref{alg:mst}] \label{thm:privacymst}
\cref{alg:mst} is $(\alpha, \alpha \rho)$-RDP for all $\alpha \geq 1$.
\end{theorem}

\begin{proof}
Step $4a$ is $(\alpha, \alpha \frac{1}{8} \epsilon^2)$-RDP by \cref{prop:expprivacy}.  Substituting $\epsilon=\sqrt{\frac{8 \rho}{r-1}}$, we see that it is equivalent to $(\alpha, \alpha \frac{\rho}{r-1})$-RDP.  It is called $r-1$ times, so the entire mechanism is $(\alpha, \alpha \rho)$-RDP by \cref{prop:composition}. 
\end{proof}}

\section{Experiments} \label{sec:experiments}

\begin{table}[t]
    \def\arraystretch{1.1}
    \centering\small
    \begin{tabular}{cc|ccc|c}
  \hline
  $\epsilon$ & Team & 3-way & High order & Income & Overall \\ 
    & & marginals & conjunctions & inequality & \\\hline
  0.3 & RMcKenna (\mech) & 0.12 & 0.17 & 0.10 & 0.13 \\
  \rowcolor{blue!10} 0.3 & \mst & \textbf{0.11} & \textbf{0.16} & 0.10 & \textbf{0.12} \\
  0.3 & DPSyn & 0.15 & 0.31 & \textbf{0.07} & 0.18 \\
  0.3 & PrivBayes & 0.19 & 0.29 & 0.18 & 0.22 \\
  0.3 & Gardn999 & 0.21 & 0.32 & 0.25 & 0.26 \\
  0.3 & UCLANESL & 0.57 & 0.72 & 0.22 & 0.50 \\
  \hline
  1.0 & RMcKenna (\mech) & \textbf{0.09} & \textbf{0.15} & \textbf{0.04} & \textbf{0.09} \\
  \rowcolor{blue!10} 1.0 & \mst & \textbf{0.09} & \textbf{0.15} & 0.05 & 0.10 \\
  1.0 & DPSyn & 0.11 & 0.23 & 0.05 & 0.13 \\
  1.0 & PrivBayes & 0.17 & 0.26 & 0.09 & 0.17 \\
  1.0 & Gardn999 & 0.18 & 0.28 & 0.22 & 0.23 \\
  1.0 & UCLANESL & 0.42 &0.53 & 0.28 & 0.41 \\
  \hline
  8.0 & RMcKenna (\mech) & \textbf{0.07} & \textbf{0.14} & 0.04 & \textbf{0.08} \\
  \rowcolor{blue!10} 8.0 & \mst & 0.08 & \textbf{0.14} & 0.05 & 0.09 \\
  8.0 & DPSyn & 0.09 & 0.20 & \textbf{0.02} & 0.10 \\
  8.0 & PrivBayes & 0.13 & 0.23 & 0.09 & 0.15 \\
  8.0 & Gardn999 & 0.17 & 0.26& 0.24 & 0.22\\
  8.0 & UCLANESL & 0.35 & 0.41 & 0.25 & 0.34 \\
  \hline
\end{tabular}
\vspace{0.5em}
\caption{\label{table:experiments} Evaluation of \mech and other mechanisms from competing teams, broken down by the three scoring metrics: 3-way marginals, high-order conjunctions, and income inequality.  \texttt{MST} is also shown for comparison, even though that mechanism was not submitted during the competition.  If it was submitted instead of \mech, it would have placed first overall.}
\end{table}

In this section we discuss the experimental evaluation carried out by the contest organizers, and how \mech compared to the submissions from other teams.  We offer our own insights into the numbers and explanations for the differences between mechanisms.

Evaluations were carried out on U.S. Census data for two different states: Arizona and Vermont.  These datasets had \num{293999} and \num{211228} records, respectively.  Scores were calculated separately for each of the three evaluation metrics described in \cref{sec:eval}.  The final score was calculated by averaging the scores for each metric, each value of $\epsilon$, and each of the two datasets.  In \cref{table:experiments} we show the score breakdown by metric (averaged over Arizona and Vermont) for the top five submitted algorithms.  We also evaluated our new mechanism (\mst) and included it as an extra row (highlighted), even though it was not evaluated by the contest organizers at the time of the competition.  This was possible because the final evaluation datasets were released after the competition, and the script used to evaluate the synthetic data was provided as part of the competitor pack.  

\mst was described in \cref{sec:extensions}.  We instantiate it with $\mathcal{C} = \set{\marginal{SEX,CITY}, \marginal{SEX,\INCWAGEA},\\ \marginal{CITY,\INCWAGEA}}$, and add $\marginal{SEX,CITY,\INCWAGEA}$ to the returned result as well.  These marginals are weighted using the same formula as \mech (see \cref{alg:select}).  

Among the contest submissions, \mech consistently performed the best, for most metrics and values of $\epsilon$.  \revision{Compared to DPSyn, it did a much better job at answering 3-way marginals and high order conjunctions, but performed slightly worse at handling the income inequality metric.   Compared to every other mechanism, \mech did better on every metric}.  

Generally speaking, \mech and DPSyn (and to a lesser extent PrivBayes) seemed to be the only mechanisms that scored well on the income inequality metric, which is surprising given that it was the simplest metric and only required preserving one 3-way marginal accurately.  This raises an important point: mechanisms that understood the evaluation criteria well, and designed their mechanisms around it, generally performed better than mechanisms that just tried to generate good synthetic data without thinking about how its utility would be evaluated.  \mech and DPSyn did a good job of designing their mechanism for the task at hand, which was an important contributing factor for why they outperformed the other solutions.

While \mech relied heavily on the provisional dataset for measurement selection, the more general variant \mst still performs well, without explicitly relying on the provisional data.  In fact, \mst would have won first place if it was submitted instead of \mech at the time of the competition;  it was only slightly worse than \mech and still better than the other submissions. \revision{The difference in performance between \mst and \mech was at most $0.01$ for every metric and privacy budget evaluated.} 

\revision{
At $\epsilon=0.3$, \mst even achieved a smaller overall score than \mech.  \mech measures both $2$- and $3$-way marginals, while \mst only measures $2$-way marginals.  Since the privacy budget is relatively small, it makes sense that fewer measurements would work better here.  A promising direction for future work is to adaptively select the number of marginals to measure based on the privacy budget and the amount of data available.}

For a much more comprehensive evaluation of these mechanisms, we refer the reader to \cite{bowen2019comparative}, which goes well beyond the metrics used in the competition to evaluate these mechanisms.  This work is also a useful resource that summarizes each mechanism at a high level.  Their results generally show that \mech was the best performing mechanism for many of the additional evaluation metrics not used as part of the official scoring criteria.  This suggests that \mech produces the most generally useful synthetic data among the submitted mechanisms.  

\section{Related Work}

Differentially private synthetic data has been an active area of research for several years.  
%
One of the earliest mechanisms proposed for this task was the ``small database mechanism'' \cite{Dwork14Algorithmic}, which instantiates the exponential mechanism over a set of small databases to select one that is statistically similar to the true data.  Unfortunately, this mechanism is not able to run in practical settings, as it requires enumerating all possible datasets of a fixed sizes, resulting in a combinatorial explosion even for small dataset sizes (especially if the domain size is large). 


In the competition, the top four submissions all followed the same basic template outlined in \cref{fig:template}: they selected and measured a collection of marginals, and then used those to estimate synthetic data.  Moreover this approach has been applied more generally in the literature \cite{zhang2017privbayes,chen2015differentially,bindschaedler2017plausible,zhang2020privsyn}. One key difference that sets our approach apart is \pgm.  While \cite{zhang2017privbayes,chen2015differentially} do leverage graphical models to generate synthetic data, their approach is limited in how it makes use of the noisy measurements: they more or less treat the noisy marginals as the true marginals (with some lightweight post-processing), whereas \pgm makes use of all available measurements to resolve inconsistencies and boost utility in a principled manner.  
An alternative method for resolving inconsistencies and generating synthetic data from noisy marginals is proposed in \cite{zhang2020privsyn}.  
This method does not construct an intermediate representation of the data distribution as \pgm does.
As a result, their consistency resolution step only ensures \emph{local consistency} (i.e., that all marginals internally agree on the common marginals), and does not satisfy the stronger notion of \emph{global consistency} (i.e., that there is a data distribution that has all stated marginals), as \pgm does.  



Another popular class of approaches for differentially private synthetic data is based on generative adversarial networks (GANs) \cite{goodfellow2014generative}.  Several differentially private GANs have been proposed for the purpose of generating synthetic data \cite{jordon2018pate,zhang2018differentially,tantipongpipat2019differentially,frigerio2019differentially,xie2018differentially,beaulieu2019privacy,abay2018privacy,torfi2020differentially,torkzadehmahani2019dp}.  In fact, two teams in the NIST competition adopted a GAN-based approach (UCLANESL in \cref{table:experiments}, and one other team that did not place in the top five), however, their scores were not generally competitive with the other approaches that used the marginal-based framework.  GANs are notoriously hard to train in practice \cite{salimans2016improved} and when differential privacy constraints are enforced, it is even more difficult.  It typically requires running an algorithm like DP-SGD \cite{abadi2016deep} to train, and if it fails to converge (which is common) the privacy budget used for training is essentially wasted.  Setting the right hyper-parameters is also a major challenge for this approach.



Another important and related problem is how to evaluate the quality of synthetic data \cite{snoke2016general,bowen2019comparative,hittmeir2019utility,arnold2020really}.  Beyond the metrics used in the NIST competition, one alternative is pMSE, which is a general measure of distributional similarity \cite{snoke2016general}.  Another alternative measure is machine learning efficacy, or how well the synthetic data supports machine learning applications \cite{hittmeir2019utility}.  A number of other measures for evaluating synthetic data can be found in \cite{bowen2019comparative}.  We believe that there may be no universal answer to this question: it should ultimately depend on the data and its use cases.  In general, it would be nice to have a mechanism that can automatically adapt to an analyst-provided \emph{workload}, and generate synthetic data that provides high utility on the queries and tasks in that workload.  Several workload-adaptive mechanisms exist, but they are generally restricted to settings where the full high-dimensional histogram can be explicitly materialized in vector form, and are thus unable to scale to high-dimensional domains \cite{mckenna2018optimizing,li2010optimizing,hardt2010simple,gaboardi2014dual,Dwork14Algorithmic}.  When combined with \pgm, the scalability (and utility) of some of these mechanisms can be improved, however \cite{mckenna2019graphical}.

\section{Conclusions}

In this paper, we described \mech, the winning mechanism from the NIST differential privacy synthetic data competition, and \mst a new mechanism inspired by \mech that works almost as well, without relying on public provisional data.  While these mechanisms are state-of-the-art, the problem of differentially private synthetic data is far from solved.  Nevertheless, we believe our basic framework centered around \pgm can serve as a core component of new mechanisms for this task.  \pgm allows the mechanism designer to focus on \emph{what to measure}, rather than \emph{how to post-process} those measurements to get synthetic data while extracting the most utility from them.  \revision{In fact, in the final round of the recently completed follow-up challenge, the NIST 2020 Temporal Map Challenge, \emph{both} the first and second place teams used \pgm for post-processing, with each team developing novel techniques for measurement selection.}

\clearpage
\bibliographystyle{abbrv}
\bibliography{refs}

\begin{appendix}
\newpage 
\section{Supplementary Material}

\revision{
\paragraph{\textbf{Noise Calibration for MST}}
We begin by calculating our total RDP budget from $(\epsilon, \delta)$ by invoking \cref{prop:rdpdp}.  In particular, we find the largest value of $\rho$ such that $(\alpha, \alpha \rho)$-RDP implies $(\epsilon, \delta)$-DP by \cref{prop:rdpdp}.  We accomplish this numerically.  We will now divide our RDP budget ``$\rho$'' equally among the three steps (4) and (5) and (6) of \mst.
To achieve $\frac{\rho}{3}$-RDP in steps (4) and (6), it suffices to set $\sigma = \sqrt{\frac{3}{2 \rho}}$ by \cref{prop:rdpgauss}.  To achieve $\frac{\rho}{3}$-RDP in step (5), we simply call \cref{alg:mst} with privacy parameter $\frac{\rho}{3}$.  The correctness of this is proven in \cref{thm:privacymst}.  The entire algorithm satisfies $\rho$-RDP by three-fold composition (\cref{prop:composition}), which translates to $(\epsilon, \delta)$-DP by \cref{prop:rdpdp}.}
\begin{figure}[h]
\includegraphics[width=0.9\textwidth]{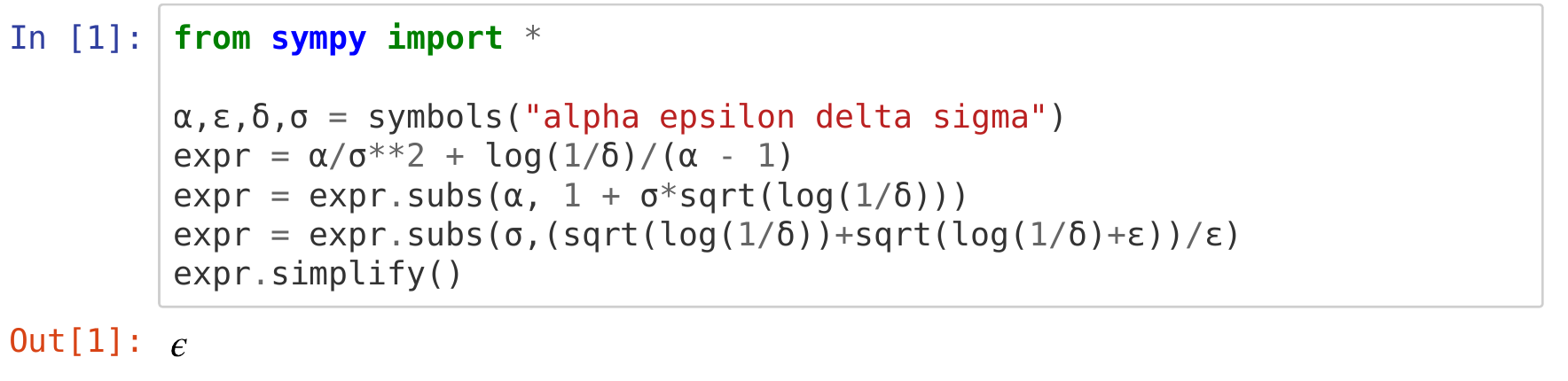}
\caption{ \label{fig:proof-privacy} \texttt{sympy} code to augment the proof of \cref{thm:privacy}}
\end{figure}
\revision{
\paragraph{\textbf{Additional Experiment on MST vs. NIST-MST}}

In \cref{sec:experiments}, we evaluated \mst in the context of the challenge test suite and evaluation metrics, and found that it performed comparably to \mech.  In this section, we compare the quality of the marginals selected from public data (\cref{alg:select}) and from the sensitive data (\cref{alg:mst}) using various privacy budgets.   Specifically, we logged the marginals selected by \mst, and compared them to the publicly chosen marginals used in \mech (ignoring the extra 2- and 3-way marginals selected in step 4) using the mutual information criteria described in \cref{alg:select}.  We also show the scores that would be achieved by the best marginals (i.e., the true maximum spanning tree), as well as the scores that would be achieved by a random spanning tree.  As shown in the table below, both the publicly chosen marginals and the privately chosen marginals nearly match the score acheived by the best marginals, for both the Arizona and Vermont datasets.  The publicly chosen marginals are slightly better at $\epsilon=0.3$ and $\epsilon=1.0$, and slightly worse at $\epsilon=8.0$.  Both variants are much better than the random baseline.

\begin{center}
\begin{tabular}{c|cc}
Selected Marginals & Arizona & Vermont \\ \hline
Best & 73.64 & 63.32 \\
Public (\cref{alg:select}) & 72.81 & 62.79 \\
$\epsilon=0.3$ (\cref{alg:mst}) & 70.53 & 61.15 \\
$\epsilon=1.0$ (\cref{alg:mst}) & 72.27 & 62.35 \\
$\epsilon=8.0$ (\cref{alg:mst}) & 73.05 & 63.09 \\
Random & 8.55 & 6.95 \\
\end{tabular}
\end{center}
}
\paragraph{\textbf{Additional Algorithms}}

\begin{algorithm}[h]
\algsize{
\textbf{Input:} $\mu$ (vector of fractional counts), n (total number of records to generate) \\
\textbf{Output:} column (synthetic column of data) \\
\nl Generate $\floor{\mu_t}$ items with value $t$ and add to column for each $t$ in domain \\
\nl Calculate remainders, $ p_t = \mu_t - \floor{\mu_t} $ \\
\nl Sample $n - \sum_t \floor{\mu_t}$ items (without replacement) from distribution proportional to $p_t$, and add to column \\
\nl Shuffle values in column \\
}
\caption{\label{alg:synthcol}Synthetic column}
\end{algorithm}

\begin{algorithm}[h]
\algsize{
\textbf{Input:} graphical model \\
\textbf{Output:} dataset (synthetic dataset) \\
\nl \revision{Initialize the set of processed attributes to the empty set} \\
\nl \For{\emph{For each attribute i}}{
\nl Let $C$ be the set of all neighbors of $i$ in the graphical model, intersected with the set of processed attributes \\
\nl \For{\emph{Group data by $C$, and for each group in $C$}}{
\nl Calculate $\mu$ from the graphical model, the vector of fractional counts for every possible value of attribute $i$, for the given group of other attributes \\
\nl Generate synthetic column for this group using \cref{alg:synthcol} \\
\nl Add this partial column to the grouped rows in the dataset \\
}
\nl \revision{Add i to the set of processed attributes}\\
}}
\caption{\label{alg:synthdata}Synthetic data generation}
\end{algorithm}

\begin{algorithm}[h]
\algsize{
\textbf{Input:} $D$ (sensitive dataset) \\
\textbf{Output:} $D$ (transformed sensitive dataset) \\
\nl Compute \attr{DIGITS} from \INCWAGEB using \cref{alg:digits} \\
\nl Replace \attr{VALUEH} and \INCWAGE attributes in $D$ using transformations: \\
\scriptsize{
\begin{minipage}{0.45\textwidth}
\begin{align*}
&\attr{VALUEH} = \begin{cases}
5 \cdot \attr{VALUEH} & \attr{VALUEH} \leq \num{5000} \\
\num{9999998} & \attr{VALUEH} = 5001 \\
\num{9999999} & \attr{VALUEH} = 5002 \\
\end{cases} \\
\end{align*}
\end{minipage}%
\begin{minipage}{0.45\textwidth}
\begin{align*}
&\INCWAGE =  
\begin{cases}
100 \cdot \INCWAGEA + \texttt{DIGITS} & \INCWAGEA \leq 50 \\
\num{9999998} & \INCWAGEA = 51 \\
\end{cases} \\ \end{align*}
\end{minipage}}}
\caption{\label{alg:reverse}Reverse transformation of \cref{alg:preprocess}}
\end{algorithm}

\begin{algorithm}[h]
\algsize{
\textbf{Input:} $k$ (a value for \INCWAGEB) \\
\textbf{Output:} $l$ (a value for \texttt{DIGITS}) \\
\nl Let $L = \set{0, \dots, 99}$ \\
\nl Let $m=[100, 20, 50, 25, 10, 5, 2, 1]$ \\
\nl \For{For $i = 0, \dots, k$}{
\nl Let $L_i = \set{ l \in L \mid l \equiv 0 \mod m_i }$ \\
\nl Let $ L = L \setminus L_i $ \\
}
\nl Sample $l$ uniformly from $L_k$ \\
}
\caption{\label{alg:digits} Convert \INCWAGEB to \texttt{DIGITS}}
\end{algorithm}

\end{appendix}

\end{document}